%% file: 2016cdc.tex
\documentclass[letterpaper, 10 pt, conference]{ieeeconf}
\IEEEoverridecommandlockouts                              
\usepackage{cite}
   \usepackage[pdftex]{graphicx}
\usepackage{array}
\usepackage{amsmath,amssymb,amsbsy,latexsym}
\usepackage[caption=false,font=footnotesize]{subfig}
\usepackage{url}
\usepackage{algorithm}
\usepackage{algpseudocode}
\usepackage{header}
\usepackage{mathtools}
\DeclarePairedDelimiter\floor{\lfloor}{\rfloor}
\newcommand\flr[2]{#1\lfloor{#2}#1\rfloor}
\newtheorem{theorem}{Theorem}
\newtheorem{proposition}{Proposition}
\usepackage{colortbl}
\begin{document}
%
\title{To Observe or Not to Observe: Queuing Game Framework for Urban Parking} 

\author{ Lillian J. Ratliff, Chase Dowling, Eric
Mazumdar, Baosen Zhang\thanks{ L. J. Ratliff, C. Dowling, and B. Zhang with Department Electrical Engineering, 
        University of Washington, Seattle, WA, USA
        {\tt\footnotesize $\{$ratliffl, cdowling,
        zhangbao$\}$@uw.edu}}\thanks{Eric Mazumdar is with the Department of
        Electrical Engineering and Computer Sciences at the University of
California, Berkeley, Berkeley, CA, USA {\tt\footnotesize mazumdar@eecs.berkeley.edu}} 
\thanks{This work is supported by NSF
 FORCES (Foundations Of Resilient
Cyber-physical Systems) CNS-1239166}
       }


\maketitle

\begin{abstract}
We model parking in urban centers as a set of parallel queues and overlay a game
theoretic structure that allows us to compare the user-selected (Nash) equilibrium to
the socially optimal equilibrium. We model arriving
drivers as utility maximizers and consider the game in which observing the queue
length is free as well as the game in which drivers must pay to observe the
queue length. In both games, drivers must decide between balking and joining. 
We compare the Nash induced welfare to the socially optimal welfare. We find
that gains to welfare do not require full information penetration---meaning, for
social welfare to increase, not everyone needs to pay to observe. Through
simulation, we explore a
more complex scenario where drivers decide based the queueing game
whether or not to enter a collection of queues over a network. 
We
examine the occupancy--congestion relationship, an important relationship for
determining the impact of parking resources on overall traffic congestion. Our simulated
models
use parameters informed by real-world data collected by the Seattle
Department of Transportation. 
\end{abstract}


%
\IEEEpeerreviewmaketitle

\section{Introduction}
\input{intro2}

\section{Queueing Framework}
\label{sec:queue}
\input{game}

\section{Free Observation Queuing Game}
\label{sec:observe}
\input{observe}

\section{Costly Observation Queuing Game}
\label{sec:unobserve}
\input{unobserve}

\section{Queue--Flow Network Simulations}
\label{sec:flow}
\input{flow}
\section{Discussion and Future Work}
\input{discussion}

\bibliographystyle{IEEEtran}
\bibliography{2016cdc}

%

\end{document}

%% file: intro2.tex
An efficient transportation system is an integral part of a well-functioning
urban municipality. Yet there is no shortage of news articles and scientific
reports showing these systems, sometimes decades old, are being stressed to
their limits~\cite{frumkin:2002aa,resnik:2010aa}. 
In recent years, \emph{congestion of surface streets} is becoming
increasingly severe and is a major bottleneck of sustainable urban growth~\cite{schrank:2015aa}.
In addition to lost productivity, there are
significant health and environmental issues associated with unnecessary
congestion~\cite{levy:2010aa,zhang:2013aa}.

A significant amount---up to 40\% in U.S. cities---of all traffic on arterials
in urban areas stems from drivers circling while looking for parking~\cite{shoup:2006aa,shoup:2003aa}. This creates an unique opportunity for municipalities to mitigate congestion. Consequently,  
the problem of \emph{smart parking} has received significant attention from both
academia and government organizations. Numerous forecasting models have been developed to predict parking availability at various timescales~~\cite{caicedo:2012aa,chen:2012aa,sasanuma:2009aa,tiexin:2012aa} and different control stategies have been proposed to keep parking occupancy at target levels~\cite{caicedo:2012aa,chen:2012aa,sasanuma:2009aa,tiexin:2012aa}. 


Pricing, both static and dynamic, is the main tool used to incentivize drivers and control the parking system. A major difficulty in developing effective pricing strategies is the \emph{asymmetry} of information between parking managers and drivers. Municipal planners do not know the complex preferences of drivers, and drivers do not known the state of the system. Therefore price signals are often ignored by the drivers, leading to inefficiencies on both sides~\cite{bolton:2005aa}. A case in point is the parking
pilot study, SFpark, conducted by the San Francisco Municipal Transportation
Agency~\cite{sfpark:aa}. In this study
drivers changed their behavior only after the \emph{second} price adjustment
because of a spike in awareness of the program~\cite{pierce:2013aa}. It
is also interesting that as coin-fed meters are replaced by smarter meters and
smartphone apps, people are actually less cognizant of the cost of
parking~\cite{glasnapp:2014aa,carney:2013aa}. This motivates a key focus of this
paper: in contrast to considering pricing as the main incentive, we study how information access impacts behaviors of drivers. 

We model parking as system of parallel queues and impose a game theoretic structure on them. Each of the queues represents a street blockface along which drivers can park. The queue itself
is abstractly modeled as the roadways and circling behavior is the process of
queueing. The parking spots along blockfaces are the servers in the queue
model. Drivers are
modeled as utility maximizers deciding whether to park based on the reward for
parking versus its cost. We consider two
game settings: in the first, drivers observe the queue length and thus, make an
informed decision as to whether they should join the queue to find parking or
balk, meaning they opt-out of parking and perhaps choose another mode of
transit. In the second case, drivers do not observe \emph{a priori} the queue
length and thus choose to balk, join with out observing, or pay to observe the
queue after which they join or balk as in the setting of the first game. 

We characterize the
Nash equilibrium in both cases to the socially optimal solution and show that
there are inefficiencies. We develop a simulation tool that investigates how different parameter combinations such as network topology and utilization (occupancy) can impact wait time (congestion) and welfare of drivers\footnote{Code is available at \url{https://github.com/cpatdowling/net-queue}}. In particular, we show that in the information limited game, the Nash equilibrium can be very far from the social optimal, especially when the utilization factor is high (e.g. a busy downtown district). This suggests that significant improvements in waiting time and congestion levels can be achieved.   

The remainder the paper is organized as follows. In Section~\ref{sec:queue}, we
outline the basic queuing framework applied to urban parking. In Sections~\ref{sec:observe}
and~\ref{sec:unobserve}, we describe the free observation and costly observation
queuing game, respectively. In the former, we examine  
congestion--limited balking rates and the impact on social welfare and in both
sections, we discuss on-street versus off-street parking. We present a
queue--flow network model in Section~\ref{sec:flow} and show through simulations
the  utilization and wait time for different Nash and socially optimal
equilibria.
We present a
comparative analysis for a variety of parameter combinations. Finally, in
Section~\ref{sec:discussion}, we make concluding remarks and discuss future
directions.

%% file: game.tex
We use an M/M/c/n queue to represent a collection of block faces that
collectively have an on-street parking supply of $c\geq 1$ (for background on queues see e.g. \cite{Gross2008}).  The number $n$ represents the
maximum number of customers in the system including those customers being served
(i.e.~parked) and those circling looking for parking. We make the following
assumptions: The arriving customers form a stationary Poisson process with mean
\emph{arrival rate} $\lambda>0$. The time that a customer parks for is assumed
to be exponential, which we model as the $c\geq 1$ parking spots \emph{serving}
customers with mean \emph{service rate} $\mu>0$. Waiting customers
are severed in the order of their arrival.


Using a standard framework for an M/M/c/n queue, we can calcuate the stationary
probability distribution for the queue length. Define the \emph{traffic
intensity} $\rho=\frac{\lambda}{c\mu}$
and let
 $Q_n(t)$ be the number of customers in the system at time $t$. Then
 $\{Q_n(t)\}_{t\geq 0}$ is a continuous time, ergodic Markov chain with state space $\{0,
\ldots, n\}$. The stationary probability distribution of having $k$ customers in
the system is given by
\begin{equation}
\textstyle  p_k(n)=\frac{d_k}{\sum_{k=0}^n d_k}, \ 0\leq k\leq n,
  \label{eq:pk}
\end{equation}
where 
\begin{equation}
  d_k=\left\{\begin{array}{ll}
    \frac{(\rho c)^k}{k!}, & 0\leq k\leq c-1\\
    \frac{(\rho c)^c}{c!}\rho^{k-c}, & k\geq c\end{array}\right.
  \label{eq:dk}
\end{equation}
We sometimes refer the number of the customers in the queue as the state of the system. Let $Z_k=X+Y_k$ be a random variable that measures the time spent in the system
when the state of the system is $k$ and where $X$ is a random variable
representing the service time and $Y_k$ is a random variable representing the
time that the customer spends in the queue. The random variables $X$ and $Y_k$
are independent, $X$ has an exponential distribution with density
$f(t)=\mu e^{-\mu t}$, and $Y_k$ (for $k\geq s$) has a gamma distribution with
density \cite{Walrand1983} 
\begin{equation}
\textstyle  g_k(t)=\frac{(c\mu)^{k-c+1}}{(k-c)!}t^{k-c}e^{-c\mu t}.
  \label{eq:Ydensity}
\end{equation}
If $h(t)$ is the waiting cost to a customer spending $t$ time units in the
system, then the expected waiting cost to a customer who arrives and finds the system in state $k$ is given by $E[h(Z_k)]$. 
While we can consider non-linear waiting cost functions, for simplicity we will
assume that it is a linear function with constant waiting cost parameter $C_w>0$,
i.e.~$h(t)=C_wt$.

In the following two sections, we consider two game theoretic formulations
overlaid on the queuing system. First, we consider the game in which arriving
customers can view the queue length and then decide whether or not to join or
balk. We refer to this game as the \emph{free observation queue game}. This setting represents a ideal situation where the entire state state information is available to all of the users, which is not currently achievable in practice. But this game is easy to analyze and serves as a useful comparison to the second game.

In the second game, we consider the setting where arriving customers do not \emph{a priori} know  the
queue length. Instead, they choose to either balk, join without knowing the
queue length, or pay a price to observe the queue after which they balk or join.
We refer to this game as the \emph{costly observation queue game}. In both these games, we study the impact of the maximum number of customers in
the system on efficiency and examine the difference between the socially optimal
and the user-selected equilibrium.


%% file: observe.tex
We first consider the observable queue game in which arriving customers know the
queue length and choose to join by maximizing their utility which is a function
of the reward for having parked and the cost of circling and paying for parking.
The \emph{nominal expected utility} of an arriving customer to the system in state $k$
is $\alpha_k=R-w_k$ where $R>0$ is the reward for parking. The \emph{total
expected utility} of a customer arriving to the system in state $k$ is given
by 
\begin{equation}
\textstyle  \beta_k=\alpha_k-\frac{C_p}{\mu}=R-\frac{C_w(k+1)}{\mu c}-\frac{C_p}{\mu}
  \label{eq:reward}
\end{equation}
where $C_p$
is the cost for parking. 
If the customer balks, the expected utility is zero. 

It can be easily verified that the sequence $\{\alpha_k\}$ is decreasing and as is
$\{\beta_k\}$. Furthermore,
the optimal strategy for a customer finding the queue in state $k$ and deciding
whether or not to join by maximizing their expected utility is to
join the queue if and only if $\beta_k\geq 0$.  In this case, if the decision to
join the queue depends on the customer optimizing their individual utility, then
the system will be a M/M/$c$/$n_b$ where
\begin{equation}
  n_b=\textstyle \flr{\Big}{\frac{R\mu c-C_pc}{C_w}}
  \label{eq:ns}
\end{equation}
is the \emph{balking level} and is
determined by solving $\beta_{n_b-1}\geq 0>\beta_{n_b}$. Let $x$ denote
the strategy of an arriving customer and suppose $x\in \{j,
b\}$ where $j$ represents \emph{joining} and $b$ represents \emph{balking}.
Hence, the equilibrium strategy
for customers is 
\begin{equation}
  x=\left\{\begin{array}{ll}
    j, & 0\leq k<n_b\\
    b, &\text{otherwise}\end{array}\right.
  \label{eq:eqstrat}
\end{equation}

The socially optimal strategy, on the other hand, is determined by maximizing
social welfare. For a M/M/$c$/$n$ queue, the total expected utility per unit time obtained by the
customers in the system is given by
\begin{equation}
  \textstyle U_{sw}(n)=\lambda \sum_{k=0}^{n-1}p_k(n)\beta_k
  \label{eq:netbenifit}
\end{equation}
\begin{theorem}[{\cite[Theorem~1]{knudsen:1972aa}}]
  There exists $n_{so}$ maximizing $U_{sw}(n)$ and $n_{so}\leq n_b$ so that
  $U_{sw}(n_b)\leq U_{sw}(n_{so})$.
  \label{thm:knudsen}
\end{theorem}

A consequence of the above theorem is that the socially optimal utility per
customer is greater than the selfishly obtained one and equality only holds
because the function $U_{sw}(n)$ is defined over $\mathbb{R}_{\geq0}$. Ideally we would like to
adjust the utility of customers to close the gap between the social optimum and
the user-selected equilibrium.
In order to obtain the socially optimal balking rate $n_{so}$ we can adjust the
price for parking $\hat{C}_p=C_p+\Delta C_p$.
\begin{proposition}
  The pricing mechanism $\hat{C}_p$ that achieves the socially optimal balking level
  $n_{so}$ is determined by solving 
$\alpha_{n_{so}}<\hat{C}_p/\mu\leq
  \alpha_{n_{so}-1}$.
  \label{prop:so}
\end{proposition}
\begin{proof}
  The goal is to find $\Delta C_p$ such that $n_{so}$ is the balking rate. Let
  the reward under the new price of parking $\hat{C}_p=C_p+\Delta C_p$ be
  \begin{equation}
    \hat{\beta}_{k}=\textstyle R-\frac{C_w (k+1)}{\mu c}-\frac{C_p+\Delta C_p}{\mu}
    \label{eq:newreward}
  \end{equation}
  We
  know that $n_{so}$ will be the balking rate if and only if
  $\hat{\beta}_{n_{so}-1}\geq 0>\hat{\beta}_{n_{so}}$. Hence, 
  \begin{align}
   \hat{\beta}_{n_{so}-1}&\textstyle =R-\frac{C_w (n_{so}-1)}{\mu c}-\frac{C_p+\Delta
    C_p}{\mu}>0\\
    &\textstyle \geq R-\frac{C_w n_{so}}{\mu c}-\frac{C_p+\Delta
    C_p}{\mu}
    \label{eq:ll}
  \end{align}
  Rearranging, we get $\alpha_{n_{so}}<\hat{C}_p/\mu\leq
  \alpha_{n_{so}-1}$.
\end{proof}
\subsection{Congestion--Limited Balking Rate}
Instead of adjusting the price of parking to close the gap between the socially
optimal solution and the user--selected equilibrium, consider the
problem of designing the balking rate to achieve a particular level of
parking--related congestion. For many municipals, congestion would be the main objective. 


In order to meet this objective, we can adjust the price of parking by selecting
$\Delta C_p$ in our game
theoretic framework so that the balking level $n_b$, being the number of cars in
the queuing system after which arriving customers decide to balk instead of
join, 
is set
to be the desired number of vehicles equaling 10\% of the total volume over the
period of interest which we denote by $n_{cl}$. 
\begin{proposition}
  The pricing mechanism $\hat{C}_p$ that achieves the congestion--limited balking level
  $n_{cl}$ is determined by solving 
$\alpha_{n_{cl}}<\hat{C}_p/\mu\leq
  \alpha_{n_{cl}-1}$.
\end{proposition}
The above proposition is proved in the same way as Proposition~\ref{prop:so};
hence, we omit it.

Note that the value of $n_{cl}$ may not be equal to $n_{so}$
since the objectives that produce these values may not be aligned. Thus,
designing the price of parking to maintain a certain level of congestion in a
city may not be socially optimal. Similar results have been shown in the
classical queuing game literature with regards to designing a toll that
maximizes revenue (see, e.g.,~\cite[Section~6]{knudsen:1972aa}).
\begin{proposition}
  Whether or not $n_{cl}\leq n_{so}$ or $n_{cl}>n_{so}$, $U_{sw}(n_{cl})\leq
  U_{sw}(n_{so})$. Furthermore, if $n_{cl}\leq n_{so}$,
 then $U_{sw}(n_{cl})= U_{sw}(n_{so})$.
\end{proposition}
  The proof of the above proposition is due to the fact that $n_{so}$ is the maximizer of
 $U_{sw}$. It tells us that selecting the balking rate to limit congestion may
 result in a decrease in social welfare.
 
\begin{proposition}
  If $n_{cl}\leq n_{b}$, where $n_b$ is the user-selected
 balking rate, then $U_{sw}(n_b)\leq U_{sw}(n_{cl})$ and vice versa.
\end{proposition}
\begin{proof}
 The result is implied by the fact that 
that $U_{sw}(n)$ is \emph{unimodal}, i.e.~$U_{sw}(n)-U_{sw}(n-1)\leq 0$ implies
that $U_{sw}(n+1)-U_{sw}(n)< 0$. Barring a little algebra, this is almost trivially true since
$\{\beta_k\}$ is a decreasing sequence; indeed,
\begin{align}
  U_{sw}(n+1)-U_{sw}(n)=&\textstyle \rho\frac{D_{n-1}}{D_{n+1}}(U_{sw}(n)-U_{sw}(n-1))\notag\\
  &\textstyle -\frac{d_{n-1}}{D_{n+1}}(\beta_{n-1}-\beta_n)
  \label{eq:unimod}
\end{align}
Since $U_{sw}(n)-U_{sw}(n-1)\leq 0$ by assumption and
$\{\beta_k\}$ is decreasing, $U_{sw}$ is unimodal. 
\end{proof}

The preceding propositions tell us that we can design the balking level
by adjusting the price to match a particular desired level of congestion, we
must be careful about how this level of congestion is selected since will impact
social welfare. In particular, selecting $n_{cl}$ will
result in a decrease in the social welfare as compared to the socially optimal
balking rate; on the other hand, it can result in an increase in social welfare if selected to
be less than the user-selected balking rate $n_b$.
\subsection{Example: Off-Street vs. On-Street Parking}
\label{sec:ex1}
Suppose customers have two alternatives. They can either choose on-street
parking by selecting to enter a M/M/$c$/$n$ queue as above with service time $1/\mu$ or they can choose
off-street parking which we model as a M/M/$\infty$ queue (infinitely available
spots) with expected service time per customer $1/\mu$. We assume the
reward $R$ is the same for both cases. The utility
for off-street parking is 
\begin{equation}
  U_{off}=\textstyle R-\frac{C_{off}}{\mu}
  \label{eq:uoff}
\end{equation}
where $C_{off}$ is the cost for off-street parking per unit time. The utility
for joining the on-street parking queue is 
\begin{equation}
  U_{on}(k)=\textstyle R-\frac{C_{on}}{\mu}-\frac{C_w(k+1)}{c\mu}
  \label{eq:uoff-1}
\end{equation}
where $C_{on}$ is the cost per unit time for on-street parking, $C_w$ is the
cost per unit time for waiting in the queue (circling for parking), and $k$ is
the state of the queue. In essence, we consider that, when a customer balks, they
choose off-street parking which represents the \emph{outside option}. Hence, we can
determine the rate at which people choose off-street in the same way as we
determined the balking rate above. In particular, we find the off-street balking
level $n_{off}$ for
which $U_{on}(n_{off}-1)\geq U_{off}>U_{on}(n_{off})$. Hence, we have that
\begin{equation}
  n_{off}= \textstyle \flr{\Big}{c\frac{C_{off}-C_{on}}{C_w}}.
  \label{eq:noff}
\end{equation}
In the sequel we will explore this example in more detail.

%% file: unobserve.tex
We now relax the above framework so that arriving customers do not observe the
queue length without paying a price. More specifically, suppose now that we have a M/M/c/n queue and that when customers arrive they can
either balk, join, or pay a cost to observe the queue length after which they
decide to balk or join. For on-street parking where there is an smart phone app
to which a customer can pay a subscription fee to gain access to information 
or
choose not to, this model makes sense. We take the theoretical model
from~\cite{hassin:2011ab}.

Assume that each customer chooses to observe the queue with probability $P_o$ at
a cost $C_o$, balks without observing with probability $P_b$, and joins with out
observing with probability $P_j$. We use the notation $P=(P_o,P_b,P_j)\in
\Delta(3)$ for the strategy of arriving
drivers where $\Delta_2=\{P=(P_o,P_b,P_j)|\ P_i\geq 0, i\in \{o,b,j\}, \
P_o+P_j+P_b=1\}$ is the strategy space, i.e.~the 2--simplex.
The \emph{effective arrival rate} for this queue is then
\begin{equation}
  \tilde{\lambda}=\left\{\begin{array}{ll}
    (1-P_b)\lambda, & k<n_b\\
    P_j\lambda, & k\geq n_b\end{array}\right.
  \label{eq:effective_lambda}
\end{equation}
where $k$ is the queue length and $n_b=\floor{\frac{R\mu c-C_pc}{C_w}}$ is the selfish balking level for the
observable case. Of course, as before, we assume that $n_b\geq 1$ to avoid the
trivial solution where $P_b=1$ is a dominant strategy. In addition, we assume
$n\geq n_b>c$ since if $c\leq n<n_b$ then users would be forced to balk $n$ and
we would just replace $n_b$ in the above equations with $n$. The only other case
is
$n<c\leq n_b$ and it is non-sensical since $c$
is the number of servers. 
We remark that if
$C_o=0$, then the game reduces to the observable game since $P_o=1$. Hence, we
investigate the case when $C_o>0$. 

The stationary probability distribution is as before (see \eqref{eq:pk}) except
we use the \emph{effective traffic intensity} $\rho=\frac{\tilde{\lambda}}{c\mu}$.
In particular, we write the balance equations
\begin{eqnarray}
  (1-P_b)\lambda p_k^n&=& (k+1)\mu p_{k+1}^n, \ 0\leq k<c \\ 
  (1-P_b)\lambda p_k^n&=& c\mu p_{k+1}^n, \ c\leq k <n_b\\
  P_j\lambda p_k^n &=&  c\mu p_{k+1}^n, \ n_b\leq k\leq n
  \label{eq:balanceunobsrv}
\end{eqnarray}
and we let $\eta=P_j\rho$, $\xi=(1-P_b)\rho$. Then, 
\begin{equation}
  p_k^n=\left\{\begin{array}{ll}
    \frac{c^k\xi^k}{k!}p_0^n, & 0\leq k <c\\
    \frac{c^c\xi^k}{c!}p_0^n, & c\leq k <n_b\\
    \eta^{k-n_b}\xi^{n_b}\frac{c^c}{c!}p_0^n, & n_b\leq k\leq n
  \end{array}\right.
  \label{eq:pk-n}
\end{equation}
so that
\begin{equation}
  p_0^n=\textstyle\left(
  \sum_{k=0}^{c-1}\frac{(c\xi)^k}{k!}+\sum_{k=c}^{n_b-1}\frac{c^c\xi^k}{c!}+\frac{\xi^{n_b}c^c}{c!}\frac{1-\eta^{n-1-n_b}}{1-\eta}
  \right)^{-1}
  \label{eq:po}
\end{equation}
Note that we now use the more compact notation $p_k(n)\equiv p_k^n$ and
similarly, we use the notation $U(n)\equiv U^n$ for utilities.

Once the customer knows the queue length
then their reward is the same as in the observable case,
i.e.~$\beta_k=R-w_k-C_p/\mu$. However, since they do not know \emph{a priori} the queue
length, the customer must make the decision as to joining, balking, or observing
by maximizing their expected utility.

 The utility for observing the queue is given by
\begin{align}
  U_o^n&(P)\textstyle=\sum_{k=0}^{n_b-1}p_k^n\beta_k-C_o\\
  &\textstyle=p_0^n\Big[ \left(R-\frac{C_p}{\mu}\right)\left(
  \sum_{k=0}^{c-1}\frac{c^k\xi^k}{k!}+\sum_{k=c}^{n_b-1}\frac{c^c\xi^k}{c!}
  \right)\notag \\
  &\textstyle- \frac{C_w}{c\mu}\left(
  \sum_{k=0}^{c-1}\frac{c^k\xi^k(k+1)}{k!}+\sum_{k=c}^{n_b-1}\frac{c^c\xi^k(k+1)}{c!}
  \right)\Big]-C_o,
  \label{eq:Uo-n}
\end{align}

the utility for joining without observing is given by
\begin{align}
\textstyle  U_j^n&(P)=\textstyle\sum_{k=0}^{n-1}p_k^n\beta_k\\
  &\textstyle=p_0^n\Big[ \left(R-\frac{C_p}{\mu}\right)\Big(
  \sum_{k=0}^{c-1}\frac{c^k\xi^k}{k!}+\sum_{k=c}^{n_b-1}\frac{c^c\xi^k}{c!}\notag\\
  & \textstyle+\sum_{k=n_b}^{n-1} \eta^{k-n_b}\xi^{n_b}\frac{c^c}{c!}\Big)
  - \frac{C_w}{c\mu}\Big(
  \sum_{k=0}^{c-1}\frac{c^k\xi^k(k+1)}{k!}\notag\\
&\textstyle+\sum_{k=c}^{n_b-1}\frac{c^c\xi^k(k+1)}{c!}+\sum_{k=n_b}^{n-1}
\eta^{k-n_b}\xi^{n_b}\frac{c^c(k+1)}{c!}
  \Big)\Big],
  \label{eq:Uj-n}
\end{align}
and the utility for balking is $U_b^n(P)\equiv0$. 
\begin{proposition}
  A symmetric, mixed Nash equilibrium exists for the game $(U_o^n, U_b^n,
  U_j^n)$.
\end{proposition}
\begin{proof}
The above proposition is a direct consequence of 
Nash's result for finite games~\cite{nash:1951aa} that states for any finite game
there exists a mixed Nash equilibrium.   
\end{proof}
On the other hand, if we were to consider a queue where $n\rar\infty$
  (i.e.~with an infinite number of players), then Nash's result would no longer
  hold. This framework is explored in the working paper~\cite{hassin:2011ab}.

Customers are assumed to be homogeneous and thus, we seek a \emph{symmetric equilibrium} which
means that it is a best response against itself. Intuitively, depending on the
relative values of the utility functions $U_b^n,U_j^n$ and $U_o^n$, we can say that an
equilibrium $(P_o,P_j,P_b)$ will satisfy the following:
\begin{subequations}\begin{eqnarray}
P_o=1, P_b=P_j=0, & U_o^n>\max\{U_j^n,U_b^n\}\label{eq:pure1}\\
    P_b=1, P_o=P_j=0, & U_b^n>\max\{U_o^n, U_j^n\}\label{eq:pure2}\\
    P_j=1, P_o=P_b=0, & U_j^n>\max\{U_o^n,U_b^n\}\label{eq:pure3}\\
    P_o=0, 0\leq P_j,P_b\leq 1, & U_b^n=U_j^n>U_o^n\label{eq:max1}\\
    P_j=0, 0\leq P_o,P_b\leq 1, & U_b^n=U_o^n>U_j^n\label{eq:max2}\\
    P_b=0, 0\leq P_j,P_o\leq 1, & U_j^n=U_o^n>U_b^n\label{eq:nash1}\\
    0\leq P_b,P_j,P_o\leq 1, & U_o^n=U_j^n=U_b^n\label{eq:nash2}
  \label{eq:conditions}
\end{eqnarray}
\end{subequations}
We adapt the best response algorithm in~\cite{hassin:2011ab} to the case where
the utility of the outside option $U_b$---which may be balking to other modes of
transit or selecting off--street parking---is not necessarily non-zero.
In particular, we use the above equations to create an
algorithm that allows us to compute the best response (see
Algorithm~1). 
In Algorithm~1, $\vep,\delta>0$ and
$\gamma\in (0,1)$. As $\vep, \delta\rar 0$, the algorithm converges to a Nash
equilibrium since its out put will approach the solution to
\eqref{eq:conditions}.  
We conjecture that the Nash equilibrium is unique and empirically observe this in the simulations. This conjecture is true when the number of players is
infinite and $U_b=0$~\cite{hassin:2011ab}.

\begin{algorithm}[h]
  \centering
\caption{Best Response Algorithm}
\small\begin{algorithmic}[1]
\Function
{getBestResponse}{$P_o,P_b,P_j$, $\vep$, $\delta$, $\gamma$}\\
$\quad$\textbf{while} {$|P_o^\ast-P_o|+|P_b^\ast-P_b|<\delta$}
\State $\ ${$U_j\leftarrow U_j^n(P_j,P_o)$, $U_o\lar U_o^n(P_j,P_o)$}\\
$\ \ \quad$\textbf{if} {$U_o>\max\{U_j,U_b\}+\vep$}:
\State $\qquad$$(P_o^\ast, P_b^\ast, P_j^\ast)\lar(1,0,0)$\\
 $\ \ \quad$\textbf{elif} {$U_j>\max\{U_o,U_b\}+\vep$}
\State $\qquad$$(P_o^\ast, P_b^\ast, P_j^\ast)\lar(0,0,1)$\\
 $\ \ \quad$\textbf{elif} {$U_b>\max\{U_o,U_j\}+\vep$}
\State $\qquad$$(P_o^\ast, P_b^\ast, P_j^\ast)\lar(0,1,0)$\\
$\ \ \quad$\textbf{elif} {$|U_o-U_b|<\vep$ \& $\min\{U_o,U_b\}>U_j+\vep$}
\State $\qquad$$(P_o^\ast, P_b^\ast,
P_j^\ast)\lar\left(P_o/(P_o+P_B),P_b/(P_o+P_b),0\right)$\\
$\ \ \quad$\textbf{elif} {$|U_j-U_b|<\vep$ \& $\min\{U_j,U_b\}>U_o+\vep$}
\State $\qquad$$(P_o^\ast, P_b^\ast,
P_j^\ast)\lar\left(0,P_b,1-P_b\right)$\\
$\ \ \quad$\textbf{elif} {$|U_j-U_o|<\vep$ \& $\min\{U_j,U_o\}>\vep+U_b$}
\State $\qquad$$(P_o^\ast, P_b^\ast,
P_j^\ast)\lar\left(P_o,0,1-P_o\right)$\\
$\ \ \ \quad$\textbf{elif} {any two $\{$$|U_j-U_b|<\vep,|U_o-U_b|<\vep$,
$|U_j-U_o|<\vep\}$}:
\State $\qquad$ $(P_o^\ast, P_b^\ast,
P_j^\ast)\lar(P_o,P_b,P_j)$\\
$\ \ \quad$\textbf{end if}\\
$\quad$\textbf{end while}\\
$\qquad$\textbf{if} $|P_o^\ast-P_o|+|P_b^\ast-P_b|\geq \delta$:
\State $\qquad$ $P_o^+\lar(1-\gamma)P_o^\ast+\gamma P_o$
\State $\qquad$ $P_b^+\lar(1-\gamma)P_b^\ast+\gamma P_b$
\State $\qquad$ $P_j^+\lar(1-\gamma)P_j^\ast+\gamma P_j$
\State $\qquad$ $(P_o,P_b,Pj)\lar (P_o^+,P_b^+,P_j^+)$
\EndFunction
\end{algorithmic}
\label{alg:Nash}
\end{algorithm}
\begin{figure}[h]
  \subfloat[\label{fig:p41-1}]{
  \includegraphics[width=1.0\columnwidth]{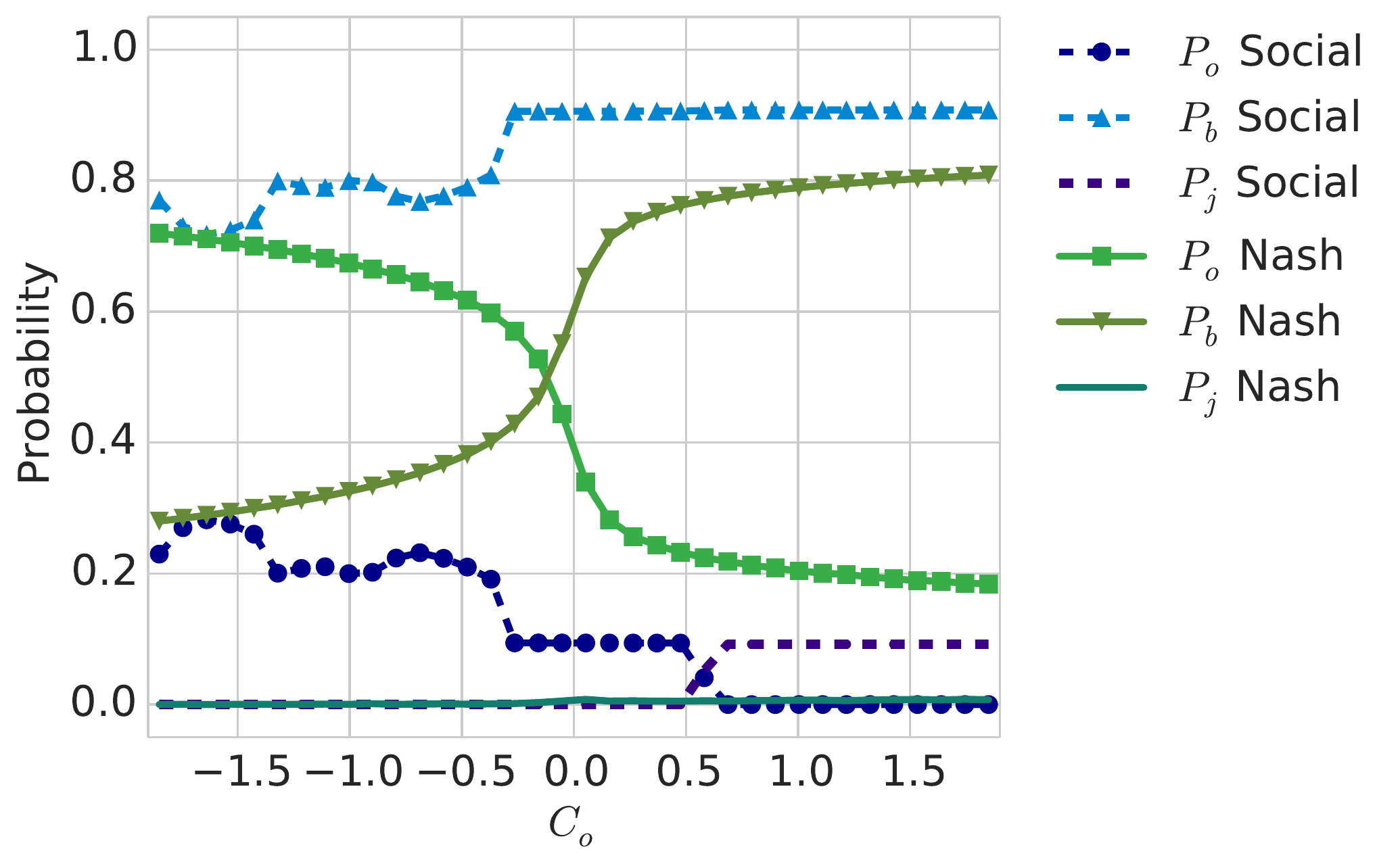}}
 
  \subfloat[\label{fig:p41-2}]{
  \includegraphics[width=0.95\columnwidth]{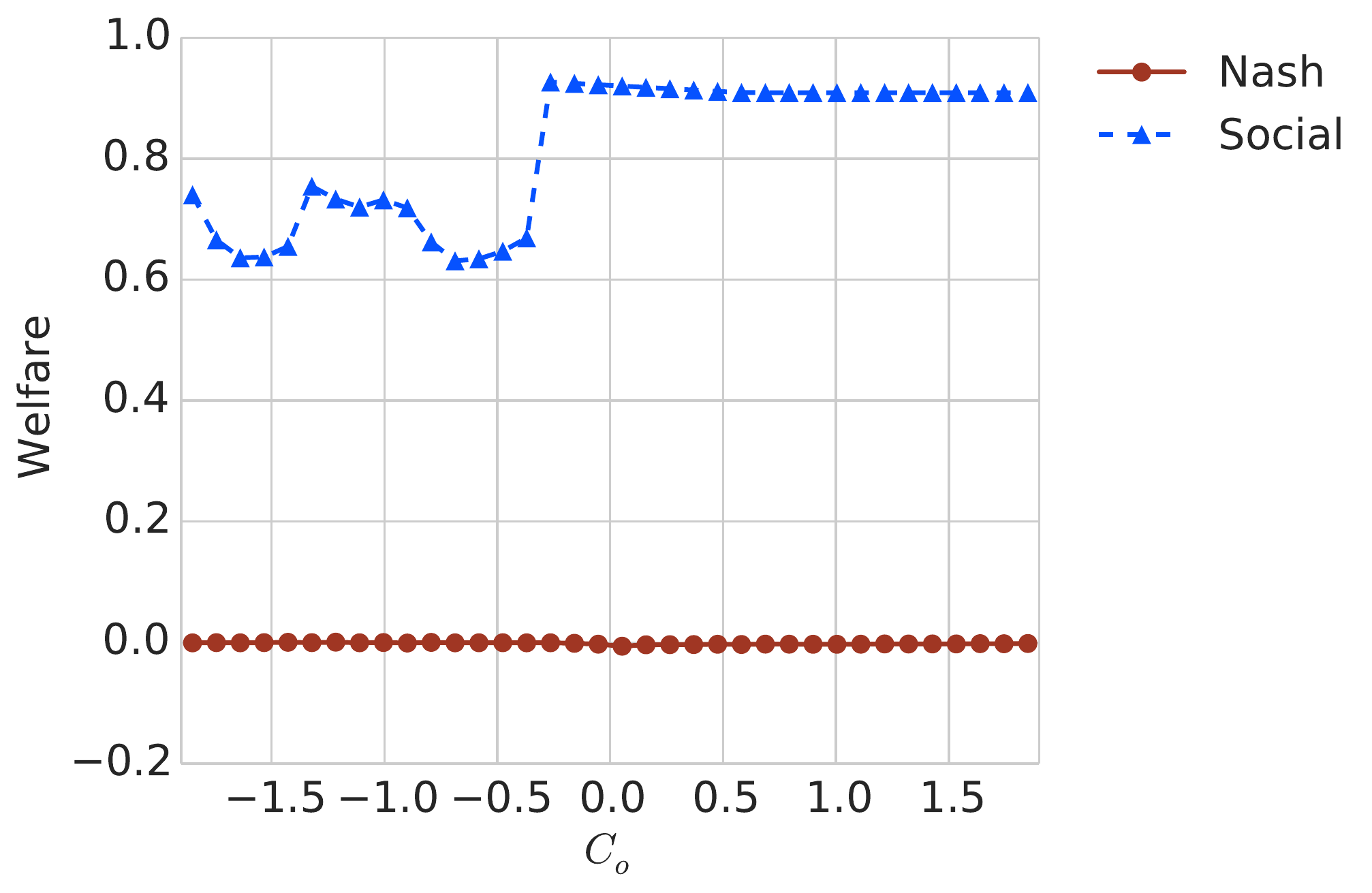}}
  \caption{(a) Nash and socially optimal equilibria for on-street vs. other modes of
  transit balk, join, observe game. (b) Social welfare and
  Nash-induced welfare. We vary $C_o$ between $-1.85$ and $1.85$ (negative
  values of $C_o$ mean the drivers are incentivized to observe) and all other parameters
  have the following values: $\lambda=1/5$, $c=30$,
  $\mu=1/120$, $C_p=0.05$, $R=75.0$, $n=100$, $C_w=1.5$. 
  }
  \label{fig:p41}
\end{figure}

On the other hand, the socially optimal strategy
$(P_o^{so},P_b^{so},P_j^{so})\in \Delta_2$
is determined by maximizing the social welfare which is given by
\begin{align}
  U_{so}^n(P)&=\textstyle P_jU_j^n(P)+ P_oU_o^n(P)+P_bU_b^n(P) \\
  &=\textstyle P_j\lambda
  \sum_{k=0}^{n-1}p_k^n\beta_k+P_o\lambda\left(\sum_{k=0}^{n_b-1}p_k^n\beta_k
  -C_o\right).
  \label{eq:socialwelfare}
\end{align}
As we stated in the previous section, it is well known that, in general, the
social welfare is not maximized by the Nash equilibrium and the Nash induced
welfare is generally less than the social welfare. For the unobservable queueing
game, we compare these the Nash and the social welfare solutions for various
parameters combinations including parameters from real-world data obtained from
the Seattle Department of Transportation. In Figure~\ref{fig:p41-2}, we show an
example of how the socially optimal welfare changes as a function of the cost of
observing $C_o$ while the Nash-induced welfare stays the same roughly the same.
However, both the Nash equilibrium and the socially optimal equilibrium vary
(Fig.~\ref{fig:p41-1})

\begin{figure}[!h]
  \center
\subfloat[ \label{fig:pos3opt-1}]{
\includegraphics[width=1.0\columnwidth]{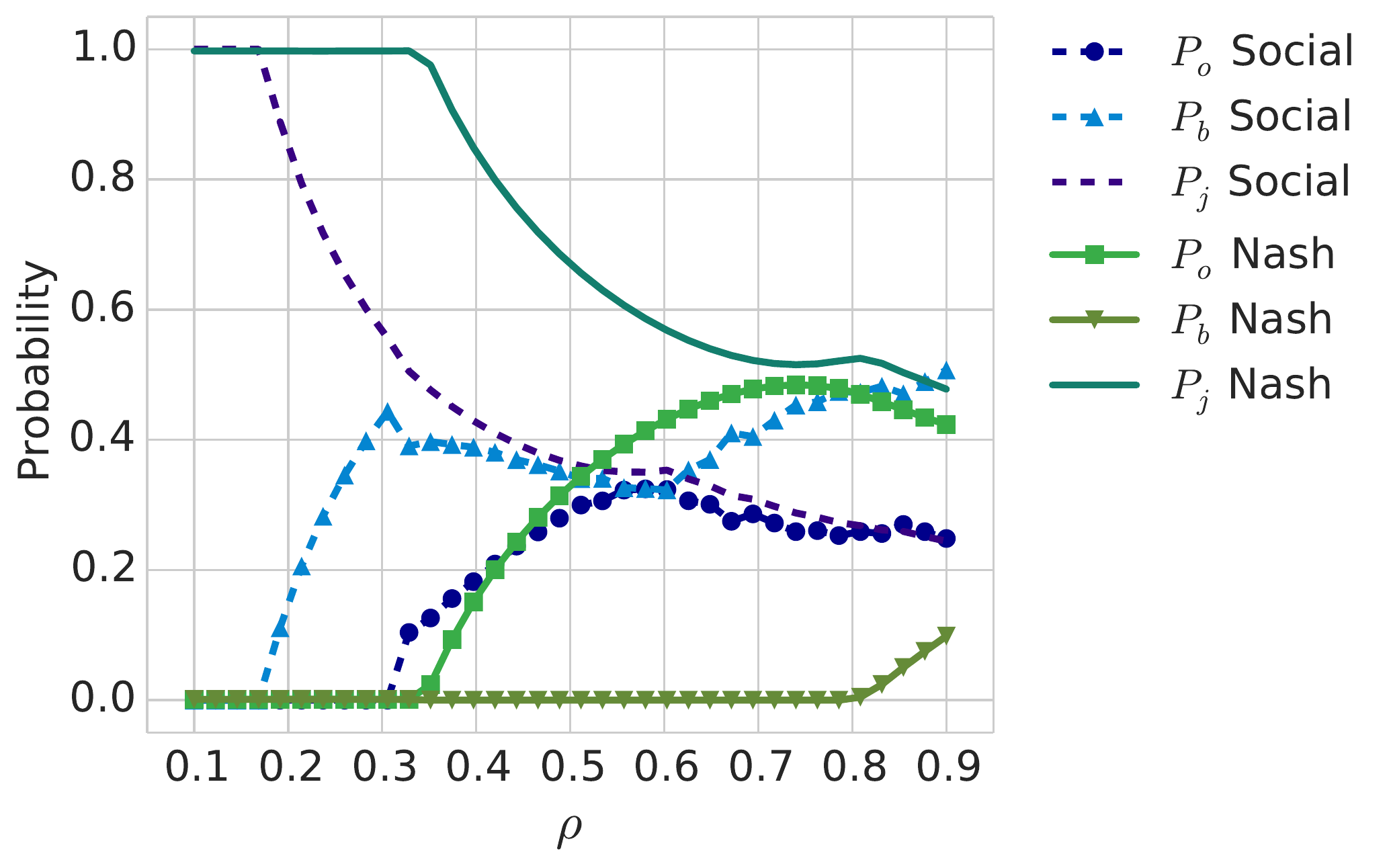}}
  
 \subfloat[ \label{fig:pos3opt}]{
 \includegraphics[width=0.95\columnwidth]{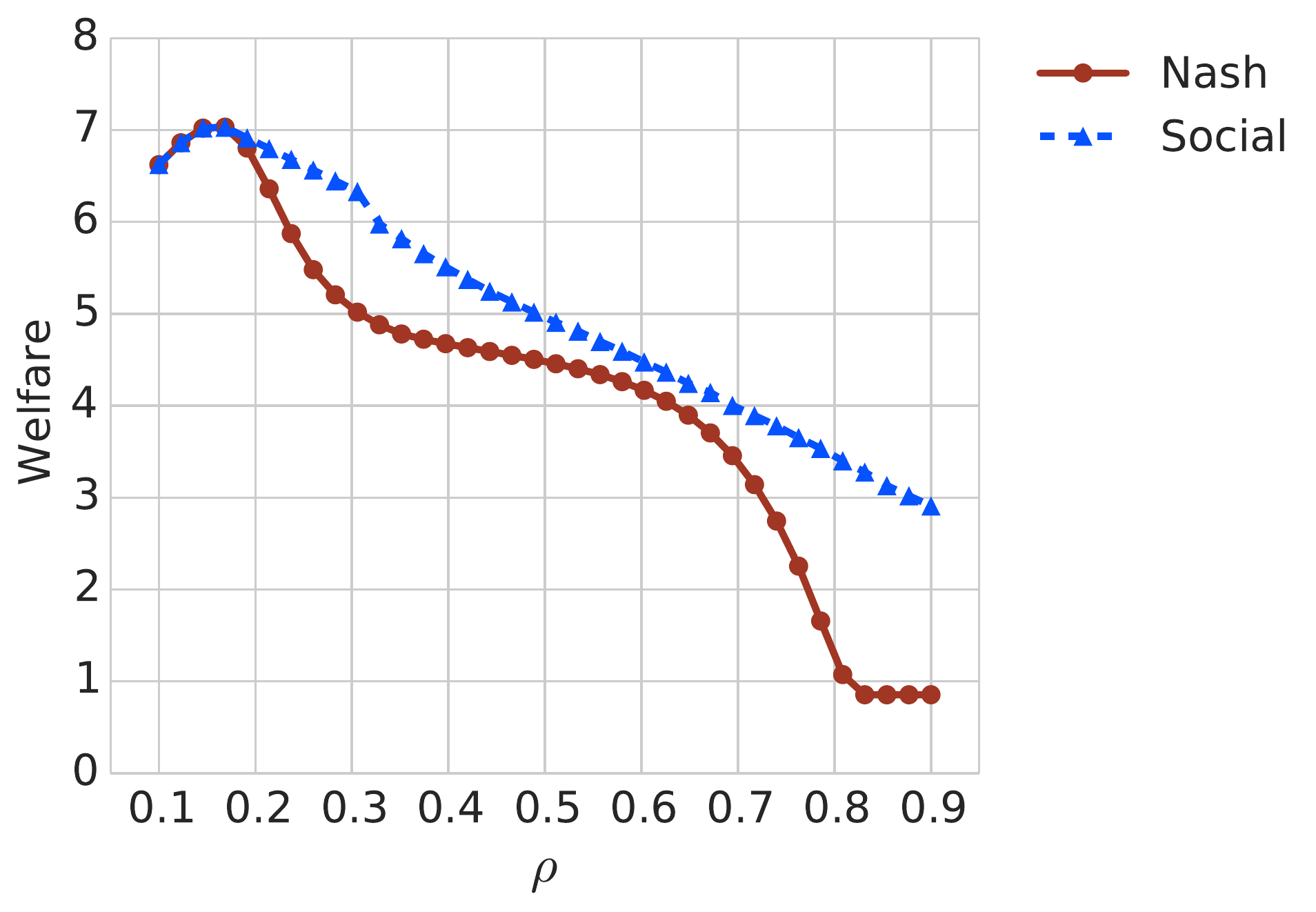}}
  \caption{(a) Nash and socially optimal equilibria. (b) Nash induced welfare vs. maximum social welfare (under socially
  optimal strategy). For both plots, the game we consider is on-street parking
  vs. off-street parking and we varied $\rho=\frac{\lambda}{c\mu}$ by keeping
  $\mu=1/120$ and $c=30$ constant and allowing $\lambda\in[0.025,0.225]$. The other parameter values are
  $C_p=0.05$, $C_o=3.85$, $R=95$, $C_w=1.5$, $n=100$, $C_{off}=0.962$.
  }
  \label{fig:pos3}
\end{figure}

 \begin{figure}[!h]
   \center
  
   \includegraphics[width=1.0\columnwidth]{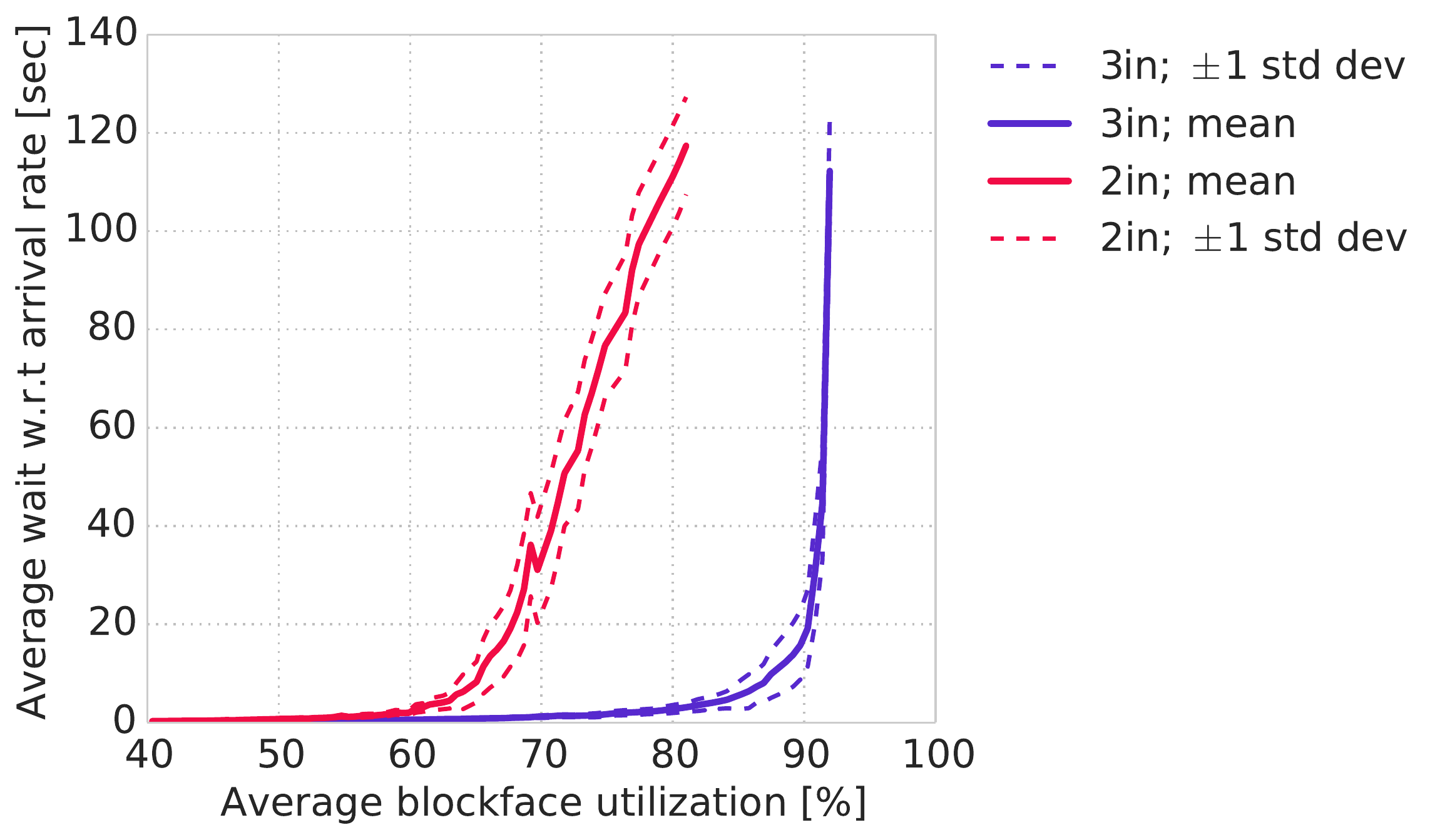}

   \caption{ 
   Average wait time (proxy for congestion) with respect to the
   exogenous arrival rate $\lambda\in[0.6,1.3]$ and fixed service rate per blockface $(c*\mu = 1.0)$ plotted against the average
   block face utilization (proxy for occupancy) for a three node queue--flow
   network with arrivals injected at all three nodes (green) and only at two
   nodes (blue). There is a distinct difference in the occupancy vs. congestion
   curves depending on the network structure and average waits grow unboundedly as $\rho \rightarrow 1$.}
   \label{fig:arrivutilwait}
  \end{figure}

\subsection{Example: On-Street vs. Off-Street Parking}
We now consider that the balking option is to select off-street parking as we
did in Section~\ref{sec:ex1}. In particular, we define
$U_b^n=U_{off}=R-C_{off}/\mu$. The Nash equilibrium can be computed using
Algorithm~\ref{alg:Nash} using $U_b^n=U_{off}=R-C_{off}/\mu$ instead of
$U_b^n=0$. On the other hand, the social welfare is now given by
\begin{align}
  U_{so}^n(P)&=P_jU_j^n(P)+ P_oU_o^n(P)+P_bU_b^n(P) \\
  &\textstyle=P_j\lambda
  \sum_{k=0}^{n-1}p_k^n\beta_k+P_o\lambda\left(\sum_{k=0}^{n_b-1}p_k^n\beta_k
  -C_o\right)\notag\\ 
  &\textstyle\qquad+P_b\lambda\left(R-\frac{C_{off}}{\mu}\right).
  \label{eq:socialwelfare_off_on}
\end{align}
In Figure~\ref{fig:pos3}, we show the Nash equilibrium and the socially optimal
strategy as well as the welfare as a function of $C_o$ for both cases for an example on-street vs.
off-street game.  

%% file: flow.tex
In this section, we present a queue-flow network model over which the games of
the previous two sections are impose. Further, we show the results of simulating
 queue-flow networks with different parameters.

  \subsection{Queue-Flow Simulator}
  Our simulator is written in Python and is freely available to download and 
  test\footnote{\url{https://github.com/cpatdowling/net-queue}}. Requirements and 
  basic instructions are available on Github. The simulator constructs a syncrhonized list of
  blockface (drivers in service) and street (drivers waiting/circling) timers linked according to the
  street topology. For simplicity the simulator treats streets and blockfaces independently: once a driver reaches the end of 
  their drive time on a street, they immediately check the entire blockface they've arrived at for availability. If no
  parking is available, the driver chooses a new destination uniformly at random based on the blockfaces currently
  accessible to them according to the street topology. High timer resolution is maintained to diminish the 
  likelihood of events occurring simultaneously and curbing potential arguments over available
  parking (e.g. a driver circling the block arriving at the same blockface as a new, 
  exogenous arrival from outside the system).
  
  In our current experimental setup, drive times between blockfaces are fixed, but could potentially be
  congestion limited, where drive time would be a function of the number of cars driving 
  on a street between blockfaces.  We consider a 3 block face system with 10 parking
  spaces each, completely connected by two-way streets. Although each blockface can be
  considered to have its own exogenous arrival rate, to facilitate the game strategies, we have
  a single source with arrival rate $\lambda$. Drivers who do not immediately choose to balk
  arrive at a uniformly random blockface where they either observe or join directly.

   \subsection{Congestion vs. Occupancy}
 The congestion--occupancy relationship is an important one to understand when
 it comes to designing the price of parking or information---e.g.~using a smart
 device with a subscription fee to observe congestion in various parking 
 zones---in order to reduce parking-related congestion. Many municipalities and
 researchers
 design pricing schemes to target a single occupancy level---typically
 \%80---for all blockfaces in a city. Not taking into account the network
 topology and node type---source or sink---can be detrimental to a pricing
 scheme. 
 In Figure~\ref{fig:arrivutilwait}, using the queue--flow simulator for a queue--flow network with
 three queues,  we show that the congestion--occupancy relationship can be
 drastically different depending on how many nodes are treated as sources for
 injections. In particular, the upper bound for utilization (before wait time
 exponentially increases) for the 3-node injection case is around 88\% while the
 2-node injection case is around 65\%. The queue--flow modeling paradigm alone  can be a useful tool for designing
discriminative pricing or information schemes, accounting for network topology.
\input{resultstab}

\subsection{Costly Observation Queuing Game Simulations}
 Coupling the queue-flow network with the game theoretic models of the previous
 sections, we simulate the queueing game and its impact on network flow (average
 wait time) and
 on-street parking utilization (occupancy).
Given a queue-flow network topology, for simplicity, we assume that each of the queues has the
same service rate $\mu$. In addition, we suppose that the total number of
parking spots (servers) across all queues in the network is $c$, the arrival rate to the
 queuing network is $\lambda$, and the capacity of queue-flow network 
is $n$. This allows us to model the whole queuing
system as a M/M/$c$/$n$ queue to which we apply the costly observation queuing game for
various parameters combinations.

We execute our simulation as follows. First, we determine the equilibrium
of the game---or the
socially optimal strategy depending on which we intend to simulate---and then, we use the simulator described above to
determine the average waiting time and utilization. The game only effects the
arrival process of the queue--flow network; once arriving drivers enter the
network, the queue--flow simulator determines the drivers
impact on the system and the waiting time they experience.

Given a strategy $(P_o, P_b, P_j)$---either a Nash equilibrium
computed via Algorithm~\ref{alg:Nash} or a social optimum computed by maximizing
\eqref{eq:socialwelfare}---we sample from the Poisson distribution with
parameter $1/\lambda$ to determine the arrival time of the next driver. Then, we
sample from the distribution determined by $(P_o, P_b, P_j)$ to decide if the
arriving driver will balk, join without observing, or pay to observe. If the driver balks, then we
discard this arriving car. If the driver joins without observing, then we
determine which node the driver enters by randomly choosing (using a uniform
distribution) a queue in the network. If the driver pays to observe, then we
examine the length of each queue in the system and the driver joins the queue
with the shortest length as long as it is less than the balking
rate. 

\begin{figure}[!t]
  \begin{center}
    \subfloat[\label{fig:pos7util}]{\includegraphics[width=0.925\columnwidth]{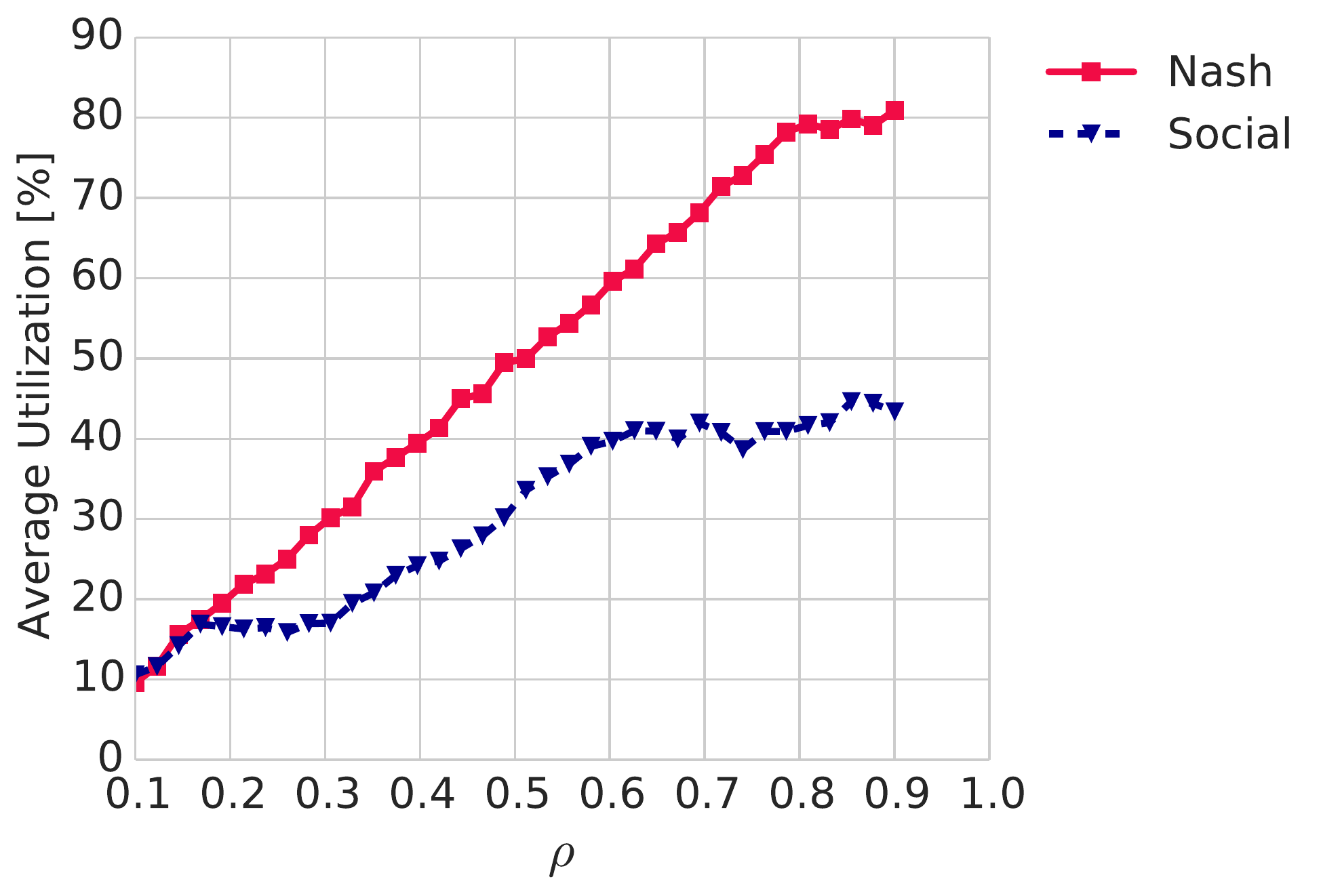}}
    \hspace{0.1in}\subfloat[\label{fig:pos7wait}]{\includegraphics[width=0.925\columnwidth]{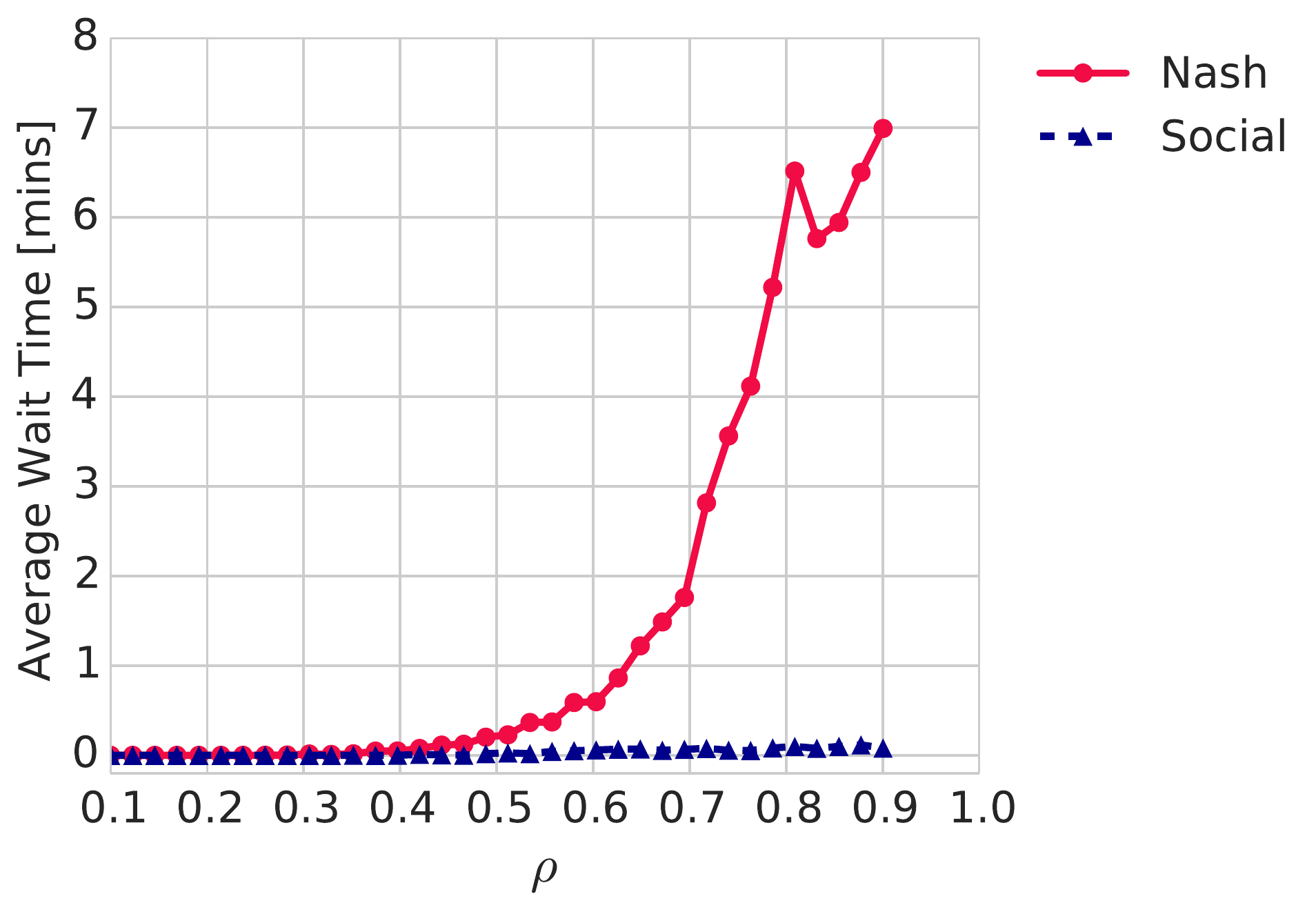}}
  \end{center}
  \caption{(a) Average wait time (a) and (b) average blockface utilization as a
  function of $\rho=\frac{\lambda}{c\mu}$ where
  $c=30$ and $\mu=1/120$ are fixed and $\lambda\in[0.025,0.225]$ for a three node system. The other parameter values are
  $C_p=0.05$, $C_o=3.85$, $R=95$, $C_w=1.5$, $n=100$, $C_{off}=0.962$. The Nash
  equilibrium and the socially optimal equilibrium varies with $\rho$ and is
  depicted in Figure~\ref{fig:pos3}. }
  \label{fig:waitutil}
\end{figure}

In Table~\ref{tab:results}, we show the results of simulations for both the
costly observation game simulations for the costly observation queuing game and the
on-street vs. off-street example. We explore different parameter
combinations and show  the utilization rate, average wait time, social welfare, and the
Nash-induced welfare. 
The social welfare is always higher than the Nash welfare,
which is to be expected. The utilization rate and average wait time are always less
under the socially optimal strategy than the Nash equilibrium. 

In Figure~\ref{fig:waitutil}, we show the result of simulating both the Nash
equilibrium and the socially optimal strategy for various values of the traffic
intensity $\rho$ (holding all other parameters fixed). These simulations are for
the same games depicted in Figure~\ref{fig:pos3}. As the traffic intensity
increases, we see that both the Nash and socially optimal utilization increase
almost linearly
with the Nash utilization remaining greater.
The socially optimal equilibrium in all cases keeps waiting times for parking---our current surrogate for congestion---uniformly less than the Nash equilibrium. Intuitively this makes sense: given a finite resource---parking---the socially optimal strategy ensures this resource is more freely available. On the other hand, the Nash strategy more efficiently utilizes the resource to the extent of its availability.
Another interesting thing to notice is the drop in wait time for the Nash solution at $\rho=0.8$. If we
look at Figure~\ref{fig:pos3opt}, we see that the probability for balking $P_b$
in the Nash solution suddenly becomes non-zero at $\rho=0.8$. This is likely to
be the cause of the drop in wait time; however, as $\rho\rar 1$, we expect the
wait time to blow-up so that after the drop, wait time continues to increase.

Of the two scenarios (costly observation and on-street vs. off-street), the
on-street vs. off-street parking more closely resembles
reality in the sense that the off-street option exists. One might guess then,
that a user will maximize their utility in either the Nash or socially optimal
case by frequently taking advantage of the ability to observe before choosing
where to park, given the nature and travel constraints of the system. What
surprises us is that only partial information availability amongst the
users---as seen in Table~\ref{tab:results}, Figures~\ref{fig:p41} and
\ref{fig:pos3} where $P_o\neq 1$ for the socially optimal solution---is required
to increase social welfare. Moreover, it seems the socially optimal equilibrium strategy requires \emph{less} information availability.

From a municipality's perspective, this is a useful result when designing a
socially optimal parking infrastructure. Not everyone will know information
about parking availability in the first place (e.g. tourists vs. residents). Even from a more practical point of view, this is a useful result in that reaching 100\% information availability for all drivers is an economically infeasible task, requiring more resources than could likely be justified.

%% file: resultstab.tex
 \begin{table*}[!t]
   \centering
   \begin{tabular}{|l|l|l|l|l|l|}
     \arrayrulecolor{white}\hline\arrayrulecolor{white}\hline
\arrayrulecolor{white}\hline
\arrayrulecolor{white}\hline
\arrayrulecolor{white}\hline
\arrayrulecolor{white}\hline
\arrayrulecolor{white}\hline
\arrayrulecolor{white}\hline
\arrayrulecolor{white}\hline
\arrayrulecolor{white}\hline
\arrayrulecolor{white}\hline
\arrayrulecolor{white}\hline
\arrayrulecolor{white}\hline
\arrayrulecolor{white}\hline

    \arrayrulecolor{black}  \hline
     \textbf{On-Street Parking vs. Other Modes of Transit} &  & Equilibrium & & &\\
     Parameters & Type &   $(P_o,P_b,P_j)$ & Utilization & Avg. Wait &
     Welfare \\
\hline
$\lambda=1/5$, $C_o=0.25$, $R=75$, $C_w=0.8$, $C_p=0.05$&SO &
$(0.00,0.58,0.42)$   &$33.2$\% & $0.002$ &  $2.80$ \\
\cline{2-6}
 &N & $(0.85,0.13,0.02)$ & $69.3$\% & $0.359$ & $0.00$  \\
 \hline
$\lambda=1/4.85$, $C_o=0.5$, $R=75$, $C_w=0.75$, $C_p=0.05$&SO &
$(0.00,0.56,0.44)$   &$34.9$\% & $0.002$ &  $3.02$ \\
\cline{2-6}
 &N & $(0.84,0.09,0.07)$ & $77.9$\% & $0.901$ & $0.00$  \\
\hline
$\lambda=1/4.5$, $C_o=2.0$, $R=75$, $C_w=0.5$, $C_p=0.075$&SO &
$(0.00,0.4,0.6)$   &$52.3$\% & $0.04$ &  $4.27$ \\
\cline{2-6}
 &N & $(0.55,0.00,0.45)$ & $88.0$\% & $3.69$ & $2.68$  \\
\hline
  \hline
     \textbf{On-Street vs. Off-Street Parking} & & &  & & \\
     Parameters &  Type&  Eq.: $(P_o,P_b,P_j)$ & Utilization & Avg. Wait &
     Welfare \\
\hline
$\lambda=1/4.5$, $C_o=3.85$, $R=65$, $C_w=1.5$, $C_{off}=0.962$,&SO& 
$(0.47,0.19,0.34)$   &$69.9$\% & $1.99$ &$6.58$ \\ 
\cline{2-6}
$C_p=0.05$ &N& $(0.49,0.00,0.51)$ & $84.0$\% & $7.77$ & $1.85$ \\
 \hline
 $\lambda=1/4.75$, $C_o=3.85$, $R=65$, $C_w=1.5$, $C_{off}=0.962$,&SO& 
$(0.5,0.14,0.36)$   &$70.6$\% & $2.23$ &$9.23$ \\ 
\cline{2-6}
$C_p=0.05$ &N& $(0.53,0.00,0.47)$ & $81.0$\% & $5.96$ & $7.19$ \\
 \hline
   \end{tabular}
   \caption{ Queue--Flow Network Game Simulation Results: for each of the
   simulations we set the total number of parking spaces to be $c=30$, the average parking
   duration is $120$ minutes ($\mu=1/120$) which is consistent with the Seattle
   Department of Transportation
   data. We use the
   shorthand SO for socially optimal and N for Nash.}
   \label{tab:results}
 \end{table*}

%% file: discussion.tex
\label{sec:discussion}
We presented a framework for modeling parking in urban environments as parallel
queues and we overlaid a game theoretic structure on the queuing system. We
investigated both the case where drivers have full information---i.e.~observe
the queue length---and where drivers have to pay to access this information.
We show in both cases that the social welfare is less under the Nash equilibrium
than the socially optimal solution and we show that only partial information
is required to increase social welfare. Finally, through simulations we connect the
queuing game to a flow network model in order to characterize wait time
(congestion) versus utilization (occupancy).

In future work, by capitalizing on the game-theoretic model, we aim to use a mechanism design
framework to shift the user-selected equilibrium to a more socially
efficient one by selecting the cost of information and the price of parking that
optimizes social welfare.
Furthermore, we plan to optimize the for the amount parking-related congestion
that
contributes to over all congestion; in particular, we plan to optimize the
social welfare as a function of the capacity of the queue. 
We plan to relax the homogeniety assumption by considering 
players with different preferences such as walking time to destination and
different priority levels such as disabled placard holders. 
Furthermore, we aim to couple the parallel queue game model with classical network flow
models for traffic flow so that we can develop an understanding of the
fundamental relationship between congestion and parking. We view the work in
this paper as the first steps toward developing a comprehensive modeling
paradigm in which the queuing behavior for parking and traffic flow are
captured.